\documentclass[11pt]{article}

\usepackage{fullpage}
\usepackage[T1]{fontenc}
\usepackage{theorem}
\usepackage{subfigure}
\usepackage{epsfig}
\usepackage{amssymb}
\usepackage{amsmath}
\usepackage{ifthen}
\usepackage{verbatim}
\usepackage{hyperref}
\usepackage{xspace}
\usepackage{bm}
\usepackage{setspace}
\usepackage{tablefootnote}
\usepackage[title]{appendix}

\newtheorem{theorem}	 			{Theorem}[section]
\newtheorem{lemma}		[theorem]	{Lemma}

\newtheorem{corollary}		[theorem]	{Corollary}
\newtheorem{prop}		[theorem]	{Proposition}

{\theorembodyfont{\rmfamily} \newtheorem{definition}
[theorem]	{Definition}}
{\theorembodyfont{\rmfamily} }
{\theorembodyfont{\rmfamily} }
{\theorembodyfont{\rmfamily} }
{\theorembodyfont{\rmfamily} }
{\theorembodyfont{\rmfamily} }
{\theorembodyfont{\rmfamily} }
{\theorembodyfont{\rmfamily} }
{\theorembodyfont{\rmfamily} }
{\theorembodyfont{\rmfamily} }
\theoremstyle{break}
{\theorembodyfont{\rmfamily} }

\usepackage{graphicx}

\newenvironment{proof}{\noindent {\em {Proof:}}}{$\blacksquare$\vskip
\belowdisplayskip}

\newenvironment{prevproof}[2]{\noindent {\em {Proof of
{#1}~\ref{#2}:}}}{$\blacksquare$\vskip \belowdisplayskip}

\newcommand{\prob}[2][]{\text{\bf Pr}\ifthenelse{\not\equal{}{#1}}{_{#1}}{}\!\left[#2\right]}
\newcommand{\expect}[2][]{\text{\bf E}\ifthenelse{\not\equal{}{#1}}{_{#1}}{}\!\left[#2\right]}

\def\eps{\epsilon}

\newcommand{\bfx}{\mathbf{x}}

\begin{document}
\title{Centralization in Block Building and Proposer-Builder Separation}
\author{Maryam Bahrani\thanks{a16z crypto. Email:
    \texttt{mbahrani@a16z.com}.} \and Pranav Garimidi\thanks{a16z
    crypto. Email: \texttt{pgarimidi@a16z.com}.} \and Tim
  Roughgarden\thanks{Columbia University \& a16z crypto. 
Author's research at Columbia
  University supported in part by NSF awards CCF-2006737
and CNS-2212745.
Email: \texttt{troughgarden@a16z.com}. 
}}

\maketitle           

\begin{abstract}
The goal of this paper is to rigorously interrogate conventional
wisdom about centralization in block-building (due to, e.g., MEV and
private order flow) and the outsourcing of block-building by
validators to specialists (i.e., proposer-builder separation):
\begin{enumerate}

\item Does heterogeneity in skills and knowledge across block
  producers inevitably lead to centralization?

\item Does proposer-builder separation eliminate heterogeneity and
  preserve decentralization among proposers?

\end{enumerate}
This paper develops mathematical models and results that offer answers
to these questions:
\begin{enumerate}

\item In a 
  game-theoretic model with endogenous staking, heterogeneous block
  producer rewards, and staking costs, we
  quantify the extent to which heterogeneous rewards lead to
  concentration in the equilibrium staking distribution.

\item In a stochastic model in which heterogeneous block producers
  repeatedly reinvest rewards into staking, we quantify, as a function
  of the block producer heterogeneity, the rate at which stake
  concentrates on the most sophisticated block producers.

\item In a model with heterogeneous proposers and specialized
  builders, we quantify, as a function of the competitiveness of the
  builder ecosystem, the extent to which proposer-builder separation
  reduces the heterogeneity in rewards across different proposers.

\end{enumerate}
Our models and results take advantage of connections to contest
design, P\'olya urn processes, and auction theory.
\end{abstract}

\section{Introduction}\label{s:intro}

\paragraph{Heterogeneity in rewards for block production.}
The economics of block production for blockchain protocols has been
growing increasingly complex over time, with block producers earning
revenue from an increasing number of different sources.  For example,
when the Bitcoin and Ethereum protocols first launched, there was
minimal demand for blockspace and minimal activity at the application
layer.  Publishing a block netted a fixed block reward for the miner
or validator that produced it but little other revenue.  In this
regime, the frequency of block production varies across block
producers (according to the hashrate or stake invested in the
protocol), but the per-block reward does not.

In time, demand for blockspace exceeded supply, forcing users to pay
non-negligible transaction fees in exchange for transaction inclusion.
Block producers were then incentivized to publish a block with the
maximum-possible sum of transaction fees.  In this case, if all
transactions are visible to all block producers (e.g., in a public
mempool), the value of a block production opportunity would remain the
same across potential block producers (namely, the block reward plus
the revenue-maximizing packing of pending transactions). If not all
block producers are aware of the same transactions, however---perhaps
because some transactions were submitted privately to one or a subset
of potential block producers---different block producers might be able
to extract differing amounts of revenue from the same block production
opportunity.

The rise of decentralized finance (``DeFi'') and the consequent
opportunities for value extraction from the application layer
(``miner/maximal extractible value,'' or ``MEV'') appears to have
exacerbated the gap in revenue that can be earned by the most and
least capable block producers~\cite{BPP,Pai_Resnick}.  A block producer with
exclusive access to high-value DeFi transactions and a proprietary and
computationally intensive algorithm to assemble them into blocks will
generally earn much more from a block production opportunity than, for
example, a hobbyist running the protocol's reference client on their
home computer .

\vspace{-.1in}

\paragraph{The centralizing forces of heterogeneity.}
Why does it matter? One concern is that heterogeneity in information
and skill across potential block producers could lead to
``centralization,'' with the blockchain protocol ultimately run by
only a very small number of the most skilled block producers (perhaps
operated by some of the world's largest financial institutions).  For
example, Buterin~\cite{endgame} writes: ``Block production is likely
to become a specialized market.''  One possible intuition for this
prediction is economic: participating in a protocol carries a cost
(e.g., in a proof-of-stake protocol, the opportunity cost of staked
capital and/or the cost of operating one or more validators) and
perhaps only those with the highest return-on-investment will find it
profitable. A different intuition stems from long-run dynamics: the
block producers that earn the highest rewards will be in a position to
increase their control over the protocol (e.g., by reinvesting rewards
as additional stake) until none of the other block producers matter.
In any case, much of the motivation behind the design of
permissionless blockchain protocols like Bitcoin and Ethereum is
exactly to avoid this type of centralization.

\vspace{-.1in}

\paragraph{Confining heterogeneity to block-building.}
How could one encourage a large set of diverse participants to
contribute to the operation of a blockchain protocol, despite what
would seem to be strong forces pushing toward centralization? One
widely-discussed idea in the Ethereum ecosystem is ``proposer-builder
separation (PBS)''~\cite{pbs_spec}, in which the role of assembling a
(presumably high-revenue) block of transactions is split out from the
role of actually participating in the blockchain protocol (validating
and proposing blocks, voting on other proposed blocks, etc.).  The
intuition is that block-building is the part of the block production
pipeline that benefits from specialization (private knowledge of
transactions, proprietary block-building algorithms, etc.), with a relatively
level playing field for everything else (publishing blocks already
assembled by builders, voting, etc.).\footnote{In the language
  of~\cite{eobp1}, PBS can be interpreted as an approach to turn
  ``active'' block producers (that care about the semantics of the
  transactions in a block) into ``passive'' proposers (that don't).  Block
  builders would then be regarded as the ``active'' participants.}
Quoting Buterin~\cite{endgame} again, on the subject of PBS:
``This ensures that at least any centralization tendencies in block
production don't lead to a completely elite-captured and concentrated
staking pool market dominating block validation.''\footnote{A separate
  concern with centralized block-building is censorship-resistance
  (i.e., preventing the potentially small number of builders from
  systematically excluding certain types of transactions). This is
  obviously an important issue, but it is outside the scope of this
  paper.}

One possible implementation of this idea, which is roughly the
implementation in Flashbots's MEV-Boost~\cite{mevboost}, is for
block-builders to submit bids along with blocks. If a block from a
builder is published by a proposer, the corresponding bid is
transferred from the builder to the proposer. (Presumably, the builder
extracts enough value from the block, via transaction fees and/or
application-layer value, to cover its bid to the proposer, perhaps
plus a premium.) The hope, then, is that builders are good enough at
their jobs that no proposer would be able to improve over the obvious
strategy of publishing the block accompanied by the highest bid.  If
all of the block proposers followed this strategy and always knew
about the full set of builder-submitted blocks, then the reward earned
by a block production opportunity would once again be independent of
the selected proposer. Ideally, restoring such homogeneity would allow
for a large and diverse set of proposers.

\vspace{-.1in}

\paragraph{Goal of this paper.} The motivation and objectives for
proposer-builder separation described above rest on plausible but
largely unsubstantiated beliefs, the rigorous interrogation of which
is the goal of this paper.
\begin{enumerate}

\item Does heterogeneity in skills and knowledge across block
  producers inevitably lead to centralization?

\begin{enumerate}

\item Is it economic forces that lead to centralization? If so, to
  what extent?

\item Do long-run dynamics lead to centralization? If so, how quickly?

\end{enumerate}

\item To what extent does proposer-builder separation reduce heterogeneity and
  preserve decentralization among proposers? How does the answer
  depend on the competitiveness of the builder ecosystem?

\end{enumerate}

\subsection{Overview of Results}

This paper develops mathematical models and results that offer answers
to all of the questions above:
\begin{enumerate}

\item Section~\ref{s:econ} addresses question~1(a) through a
  game-theoretic model with endogenous staking, heterogeneous block
  producer rewards, and staking costs.  Building on connections to
  Tullock contests, the main result (Theorem~\ref{t:BP_comp})
  quantifies the extent to which heterogeneous rewards lead to
  concentration in the equilibrium staking distribution.

\item Section~\ref{s:prob} investigates question~1(b) using a
  stochastic process in which heterogeneous block producers repeatedly
  reinvest rewards into staking. Building on connections to P\'olya
  urns and Yule processes, the main result (Theorem~\ref{t:bounds})
quantifies, as a function of the block producer heterogeneity, the
rate at which stake concentrates on the most sophisticated block
producers.

\item Section~\ref{s:pbs} studies question~2 in a model with
  heterogeneous proposers and specialized builders, with proposers
  optionally constructing their own blocks on the side. Building on
  connections to auction theory, the main result (Theorem~\ref{t:pbs})
  quantifies, as a function of the competitiveness of the builder
  ecosystem, the extent to which PBS reduces the heterogeneity in
  rewards across different proposers.

\end{enumerate}
Our results clarify the extent to and the assumptions under which the
conventional wisdom around centralization in block-building and
PBS is correct. For example, our analysis in Section~\ref{s:econ}
shows that economic forces generally lead to an oligopolistic
equilibrium outcome rather than a naive ``winner-take-all''
scenario. For another example, our analysis in Section~\ref{s:pbs}
shows that conventional intuition around PBS breaks down if the
distribution of block values is sufficiently heavy-tailed.

Our results also provide quantitative predictions that would be
impossible without concrete mathematical models. For example, our
analysis in Section~\ref{s:econ} suggests that if, say, there at
least~10 block producers that are at least 90\% as good at extracting
value as the most sophisticated block producer, then at equilibrium no
block producer will control more than roughly 17.5\% of the stake.
For another example, our analysis in Section~\ref{s:prob} suggests
that even modest constant-factor decreases in the performance gap
between the best and second-best block producers can greatly slow down
the rate of stake concentration.

\subsection{Related Work}

\paragraph{Tullock contests.}
Outside the world of blockchains, considerable effort has gone into
understanding the equilibria of games in which strategic agents must
invest resources to compete for a fixed prize. This type of game,
termed a \emph{Tullock Contest}, was first studied by Tullock
\cite{tullock} for the case of homogeneous agents. That model was then
extended by Hillman and Riley \cite{HR} and Gradstein
\cite{gradstein1995} for the case with heterogeneous agents that faced
different costs of investment. In a separate but closely related line
of work, Johari and Tskitlis \cite{JT04} and Rougharden~\cite{icm}
study equilibria in resource allocation games. Here there is a fixed
amount of resource to be distributed and agents make bids for
different shares. These models have since been ported to blockchains
in \cite{ArnostiW22,Dimitri17,alsabah2020pitfalls,Buddish18} to
understand how much miners invest in hardware and energy at
equilibriun in proof of work blockchains. In particular, Arnosti and
Weinberg \cite{ArnostiW22} and Alsabah and Capponi
\cite{alsabah2020pitfalls} demonstrate that, at equilibirum, the
market share of block production becomes centralized amongst the
miners that can invest at the cheapest cost.  One difference between
these works and our analysis in Section \ref{s:econ} is that they focus
on heterogeneity in cost of investment while keeping the reward fixed
for all miners, while our model has block producers that earn 
different rewards but with identical costs. A second is our emphasis
on parameterized definitions of and sufficient conditions for
decentralization (Definition~\ref{d:gammak} and
Theorem~\ref{t:BP_comp}).

\paragraph{Proof of stake.}
Moving to proof of stake protocols,
there have been many works investigating possible
\emph{rich get richer} phenomena,
the worry being that block producers who start with a large
fraction of the stake will grow their advantage and eventually control
all of the stake. Work by Rosu and Saleh \cite{Rosu21} shows that in a
model in which all block producers earn the same rewards, with these
rewards relatively small compared to the absolute amount of stake, 
the ratio of stake controlled by different block producers follows a
martingale, typically keeping them at essentially the same ratio over
the long run. Additional work by Huang et al.~\cite{rich_get_richer}
and Tang~\cite{tang2023trading,tang2022stability} reach similar
conclusions in somewhat different models of proof of stake protocols. 
Fanti et al. \cite{FantiKORVW19} introduce a notion of equitability that quantifies whether or not block producers maintain the same share of stake over time and they study how different reward functions impact this metric. 
In all of these works, the reward earned
by a block production opportunity is independent of the block producer
(i.e., there is no block producer heterogeneity, other than their
current stake amounts).

Our work here shows that, by contrast, with heterogeneous block producers,
a \emph{rich get richer} effect does indeed occur, with the
block producers capable of earning the highest rewards (per block
production opportunity) eventually controlling a dominant
fraction of the total stake.

\paragraph{P\'olya urns.} 
P\'olya urns are important for our analysis in
Section~\ref{s:prob} of the long run behavior of the
staking process. 
As noted in several earlier works (without BP
heterogeneity)~\cite{FantiKORVW19,rich_get_richer,Rosu21,tang2022stability,tang2023trading}, there is a strong resemblance between the
P\'olya urn setup of repeatedly drawing a ball randomly from
an urn and replacing that ball with more balls of the same color
and a process in which block producers repeatedly reinvest all earned
rewards into staking.
P\'olya urns were
first formally introduced by Eggenberger and P\'olya in \cite{polya}
and there is now a large literature on the topic.
Of particular interest
to us is a technique 
of Athreya and Karlin \cite{athreya1968embedding}
for embedding a discrete P\'olya scheme into a
continuous P\'olya process.
Janson \cite{janson} built on this
work by using the continuous-time process to prove results about the
limiting distribution for a wide variety of P\'olya urn models,
including models with color-specific replacement parameters that
correspond to the non-uniform reward rates of heterogeneous BPs in our
model.
Translated to our model, one of Janson's results implies that a block
producer with a consistent advantage in earning rewards over the other
block producers will eventually, in the limit, control all of the stake.
Our analysis here builds on Janson's result to quantify exactly how
quickly, as a function of the size of the BP's advantage, we can
expect this concentration to occur.

\section{The Basic Model}\label{s:model}

The focus of this work is on the centralizing effects of block
producer heterogeneity. To isolate this issue, this section defines
what is arguably a minimal model that allows for its study.

We consider a finite set~$I=\{1,2,\ldots,n\}$ of {\em block producers
  (BPs)}. Each BP~$i \in I$ is characterized by a {\em reward
  multipler~$\mu_i \ge 1$}; BPs with higher $\mu_i$'s earn more from
producing a block than those with lower $\mu_i$'s.
%
%
%
Specifically, we consider a block production opportunity (e.g., a slot
in the Ethereum protocol) with ``base reward'' equal to some
value~$r > 0$ and assume that, if the BP~$i$ is the one selected to
take advantage of this opportunity (e.g., by winning a proof-of-stake
lottery for that slot), then it earns $\mu_i r$ from it.  We do not
model any details of the source(s) of these rewards, which could
include a block reward, transaction fees (as computed by some
transaction fee mechanism), and/or value derived from the application
layer (i.e., ``MEV'').  We assume that a BP is chosen for a block
production opportunity with probability proportional to the amount of
the blockchain's native currency that it has staked at that
time. Thus, the expected reward earned by BP~$i$ for a given block
production opportunity is
\[
\frac{\pi_i}{\sum_{j \in I} \pi_j} \cdot \mu_i \cdot r,
\]
where the~$\pi_i$'s denote the BPs' current stakes.
We can assume, without loss of generality,
that~$\mu_1 \ge \mu_2 \ge \cdots \ge \mu_n = 1$ (reindexing and
redefining~$r$, as necessary).  The largest multiplier~$\mu_1$ can
then be interpreted as one simple measure of the ``degree of BP
heterogeneity.''

Asking ``does BP heterogeneity lead to centralization?'' can then be
investigated mathematically through questions of the form ``is the
staking distribution ultimately dominated by the BPs with the largest
$\mu_i$'s?''  Section~\ref{s:econ} probes this question
game-theoretically, with the ``ultimate staking distribution''
referring to the equilibrium stake distribution in a one-shot game
with endogenous staking. Section~\ref{s:prob} studies the question
from a dynamic (but non-strategic) perspective, with the ``ultimate
staking distribution'' corresponding to the long-run distribution of
BPs' stake following a long sequence of block production
opportunities.

\vspace{-.1in}

\paragraph{Discussion.}
In this paper, we take the $\mu_i$'s as given and deliberately avoid
microfounding the reason(s) for BP heterogeneity. Plausible reasons
why one BP might have a higher multiplier than another include:
\begin{enumerate}

\item One BP may know about more transactions than another (e.g.,
  transactions submitted privately to it rather than to the public
  mempool), and is therefore in a position to earn higher transaction
  fees and/or additional value from the application layer.

\item One BP may have a better block-building algorithm than another
  and is therefore capable of constructing higher-reward blocks.

\item One BP may be better positioned to profit from the transactions
  in a given block than another (e.g., depending on long or short
  positions held by the BP on a centralized exchange for assets traded
  in those transactions).

\end{enumerate}

In the basic model considered here (and in Sections~\ref{s:econ}
and~\ref{s:prob}), we treat a block producer as a single entity acting
unilaterally.  In practice, especially in the Ethereum ecosystem,
block production can involve ``searchers'' (who identify opportunities
for extraction from the application layer), ``builders'' (who assemble
such opportunities into a valid block), and ``proposers'' (who
participate directly in the blockchain protocol and make the final
choice of the published block).  One interpretation of a block
producer in our basic model is as a vertically integrated searcher,
builder, and proposer (e.g., as in the Ethereum ecosystem circa~2020).
Section~\ref{s:pbs} considers the ramifications of ``proposer-builder
separation,'' the more contemporary scenario in which proposers and
builders (or more precisely, integrated searcher-builders) are
separate entities.

\section{Economic Forces Toward Centralization}\label{s:econ}

\subsection{Competition Between Heterogeneous Block Producers}

To investigate economic forces toward centralization, we next extend
the basic model in Section~\ref{s:model} into a game of complete
information in which utility-maximizing BPs compete for rewards
through their choice of stake amounts. We assume that there is a fixed
(per-unit) cost of staking~$c$ (e.g., due to the opportunity cost of
capital and/or the operating costs of running one or more validators),
and that BPs act to maximize expected rewards earned less costs
incurred.  Thus, the strategy of a BP~$i$ is its choice of stake
amount~$\pi_i$, and its utility function is
\[
U_i(\pi_i; \pi_{-i}) = \mu_i \cdot r \cdot x_i(\pi_i; \pi_{-i}) - c
\cdot \pi_i,
\]
where $x_i(\pi_i; \pi_{-i}) = \pi_i/\sum_{j \in I} \pi_j$ denotes
$i$'s fraction of the overall stake.  Observe that staking offers
diminishing returns: for every non-zero nonnegative vector~$\pi_{-i}$,
the winning probability~$x_i$ (and hence the utility function~$U_i$)
of BP~$i$ is strictly concave in~$\pi_i$. 

An {\em equilibrium} of this game is then a vector
$\hat{\pi}=(\hat{\pi}_1,\ldots,\hat{\pi}_n)$ of staking amounts such
that every BP~$i$ chooses a best response to the
strategies~$\hat{\pi}_{-i}$ of the other BPs:
\begin{equation}\label{eq:eq}
\hat{\pi}_i \in \arg\max_{\pi_i \ge 0} U_i(\pi_i; \hat{\pi}_{-i}).
\end{equation}
Because of strict concavity, for every non-zero and nonnegative
vector~$\pi_{-i}$, the maximizer on the right-hand side
of~\eqref{eq:eq} is unique.  The diminishing returns to staking are
the fundamental reason why the equilibrium will be more complex
than a simple ``winner-take-all'' outcome, with multiple BPs staking
(different) non-zero amounts.

One advantage of this formalism is that it connects directly to a
well-studied economic model known as a {\em Tullock
  contest}~\cite{tullock} and its generalization to heterogeneous
preferences by Hillman and Riley~\cite{HR}.  For
example, it is known that there is always a unique equilibrium in this
model. One way to see this fact is through the following result, due
to Johari and Tsitsiklis~\cite{JT04} in an equivalent model, which
connects our model to the theory of ``potential games''~\cite{MS,icm}
by characterizing equilibria as the maximizers of a strictly concave
optimization problem.

Precisely, call a vector~$\hat{\bfx}=(\hat{x}_1,\ldots,\hat{x}_n)$ an
{\em equilibrium allocation} if it is induced by some equilibrium
staking vector (i.e., for some equilibrium~$\hat{\pi}$, $\hat{x}_i =
x_i(\hat{\pi}_i, \hat{\pi}_{-i})$ for every~$i \in I$).  Then:
\begin{prop}[Characterizing Equilibria as Optima~\cite{JT04}]\label{prop:jt}
For every sequence $\mu_1 \ge \mu_2 \ge \cdots \ge \mu_n = 1$ of
reward multipliers, every base reward~$r$, and every cost parameter~$c$, a vector
$\hat{\bfx}$ is an equilibrium allocation if and only if it is a
solution to the following optimization problem:
\begin{equation}\label{eq:opt1}
\max \sum_{i \in I} \mu_i \cdot \frac{r}{c} \cdot x_i \cdot \left(1 -
  \frac{x_i}{2} \right)
\end{equation}
subject to
\begin{align}\label{eq:opt2}
\sum_{i \in I} x_i & =1\\ \label{eq:opt3}
x_i & \ge 0 \qquad\quad\text{for all $i \in I$.}
\end{align}
\end{prop}
The proof is mechanical (see~\cite{JT04} for details): the first-order
conditions for the best-response problems~\eqref{eq:eq}, when
translated from staking vectors to allocation vectors, match the
first-order optimality conditions for the optimization
problem~\eqref{eq:opt1}--\eqref{eq:opt3}.  Because the optimization
problem is strictly concave, these conditions also characterize its
(unique) global optimum.  (This optimum exists because the
optimization problem has a continuous objective function and a compact
feasible region.)  We can therefore write $\hat{\bfx}(\mu,r,c)$ for
the unique equilibrium allocation with multipliers~$\mu$, base
reward~$r$, and cost parameter~$c$.  
Uniqueness of the equilibrium
allocation~$\hat{\bfx}$ also easily implies the uniqueness of the
equilibrium staking vector~$\hat{\pi}$.

With this setup, and identifying ``centralization'' with concentration
of the equilibrium staking distribution, we now have a concrete way to
quantify the extent to which heterogeneity across competing BPs leads
to centralization: as a function of the heterogeneity in the reward
multiplier sequence (the~$\mu_i$'s), how large are the largest
components of the corresponding equilibrium allocation vector (the
$\hat{x}_i$'s)?

\subsection{Quantifying the Largest Market Share}

How much does heterogeneity matter? If all the $\mu_i$'s are the same,
we would expect all the equilibrium allocations to also be the same
(with each BP staking a~$1/n$ fraction of the total). If the $\mu_i$'s
are different, we might expect BPs with higher $\mu_i$'s to be more
motivated to participate in staking and wind up
with larger equilibrium allocations (and indeed, this follows easily
from Proposition~\ref{prop:jt}). But how big does the variance in the
$\mu_i$'s need to be before the first BP acquires, at equilibrium, a
concerning fraction of the overall stake?

The following definition provides one parameterization of how much better
the best BP is at reward extraction than the rest.

\begin{definition}[$(\gamma,k)$-competitive]\label{d:gammak}
  A block producer set is {\em $(\gamma,k)$-competitive} if
  $\mu_{k+1} \ge \gamma \cdot \mu_1$, or, equivalently, if there are at
  least $k$ block producers that have a reward multiplier that is at
  least a $\gamma$ fraction of the largest multiplier.
\end{definition}

The main result in this section characterizes the largest-possible
equilibrium allocation of a BP, parameterized by~$\gamma$
and~$k$.\footnote{For each~$k \in
  \{1,2,\ldots,n-1\}$, there is some maximum~$\gamma_k$ for which a BP
  set (of size~$n$) is~$(\gamma_k,k)$-competitive (namely,
  $\gamma_k = \mu_{k+1}/\mu_1$).
  Theorem~\ref{t:BP_comp}(a) applies to all~$n-1$ of these choices
  for~$(\gamma_k,k)$, and in particular to the choice that minimizes
  the upper bound $1 - \gamma \cdot \tfrac{k}{k+\gamma}$.}
Recall that $\hat{\bfx}(\mu,r,c)$ denotes the unique equilibrium
allocation (i.e., market shares) guaranteed by
Proposition~\ref{prop:jt}.

\begin{theorem}[Characterization of the Maximum Market Share]\label{t:BP_comp}
Fix a base reward~$r$ and a staking cost~$c$.
\begin{itemize}

  \item [(a)] For every $\gamma \in [0,1]$ and $k \ge 1$,
for every $(\gamma,k)$-competitive BP set with multipliers $\mu_1 \ge
\mu_2 \ge \cdots \ge \mu_n = 1$,
\[
\hat{x}_i(\mu,r,c) \le 1 - \gamma \cdot \frac{k}{k+\gamma}
  \]
for every BP~$i \in I$.

  \item [(b)] For every $\gamma \in [0,1]$ and $k \ge 1$,
there exists a $(\gamma,k)$-competitive BP set with multipliers $\mu_1
\ge \mu_2 \ge \cdots \ge \mu_n = 1$ such that
\[
\hat{x}_1(\mu,r,c) = 1 - \gamma \cdot \frac{k}{k+\gamma}.
  \]
\end{itemize}
\end{theorem}
For example, if~$\gamma=1$---and so there is a $(k+1)$-way tie for the
maximum multiplier---no BP is responsible for more than a~$1/(k+1)$
fraction of the total stake at equilibrium.
If~$\gamma = \tfrac{1}{2}$ and~$k=1$, the BP most capable of reward
extraction could have as much as 67\% of the overall stake at
equilibrium; if~$\gamma=\tfrac{1}{2}$ and~$k$ is large, that BP might
control as much as (slightly more than) 50\% of the stake.  In
general, the maximum equilibrium allocation is small 
if and only if there are several BPs
that are nearly as capable as the best BP at reward extraction (e.g.,
10 BPs with multipliers within 90\% of the maximum would guarantee a
maximum market share of 17.5\%); see also Figure~\ref{f:econ}.
The bad news, then, is that even a
modest amount of heterogeneity across the highest-skill BPs can lead
to a worrying concentration of stake. The good news is that, if there
were some way to severely limit the variation in reward multipliers
(the topic of Section~\ref{s:pbs}), then, as long as there are at
least a handful of BPs, decentralization would be approximately
preserved.

\begin{figure}
\begin{center}
\includegraphics[width=.75\textwidth]{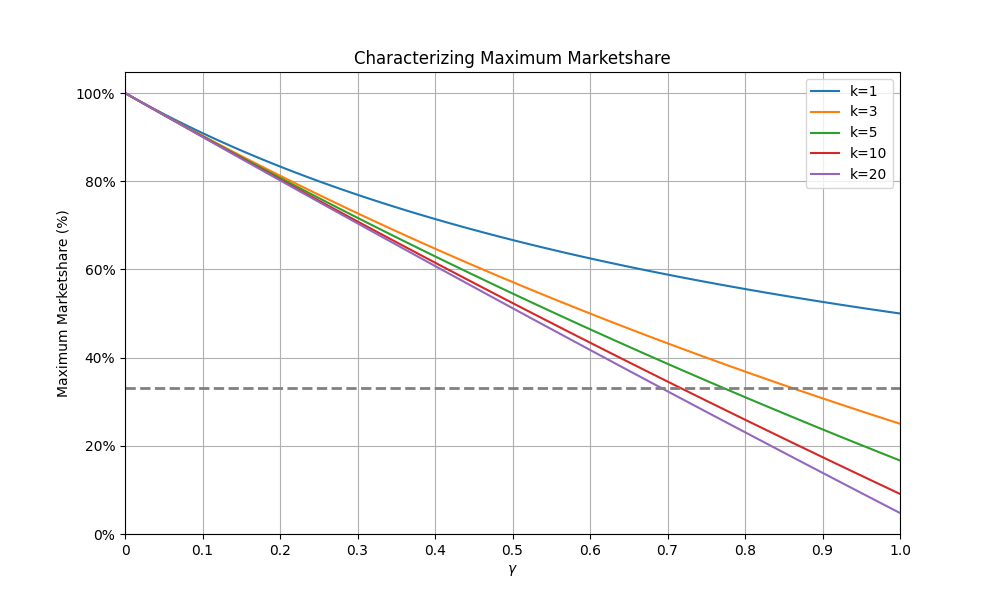}
\caption{The maximum equilibrium market share of any BP with a BP set
  that is $(\gamma,k)$-competitive in the sense of
  Definition~\ref{d:gammak}, as characterized by Theorem~\ref{t:BP_comp}.
The horizontal dotted line corresponds to a market share of 33\%,
which is a critical security threshold for many proof-of-stake blockchain protocols.}\label{f:econ} 
\end{center}
\end{figure}

We next turn to the proof of Theorem~\ref{t:BP_comp}. We require the
following monotonicity lemma, which follows from basic properties
of separable concave maximization problems (like the one in
Proposition~\ref{prop:jt}). The details are left to the appendix.

\begin{lemma}[First Monotonicity Lemma]\label{l:mono1}
For a BP set~$I$, base reward~$r$, and cost parameter~$c$, let~$\mu$
and~$\mu'$ denote two nonnegative reward multiplier
vectors that are identical except that, for some~$j \in I$, $\mu'_j >
\mu_j$.  Then, for every~$i \neq j$, $\hat{x}_i(\mu',r,c) \le \hat{x}_i(\mu,r,c)$.
\end{lemma}

With this lemma in hand, we can now prove Theorem~\ref{t:BP_comp}.

\vspace{.1in}

\begin{prevproof}{Theorem}{t:BP_comp}
We begin with part~(b). Fix~$r > 0$, $c > 0$, $\gamma \in [0,1]$, and
$k \ge 1$.  Consider a set of~$k+1$ BPs for which~$\mu_1 = 1/\gamma$
and $\mu_2 = \mu_3 = \cdots = \mu_{k+1} = 1$; this set is
$(\gamma,k)$-competitive.  We proceed to guess and
check the equilibrium allocation or, equivalently, the optimal
solution to~\eqref{eq:opt1}--\eqref{eq:opt3}. Set
\[
\hat{x}_1 = 1 - \gamma \cdot \frac{k}{k+\gamma}
\]
and
\[
\hat{x}_i = \gamma \cdot \frac{1}{k+\gamma}
\]
for each~$i=2,3,\ldots,k+1$.  This allocation is feasible (i.e.,
satisfies~\eqref{eq:opt2} and~\eqref{eq:opt3}) and satisfies the
optimality conditions in~\eqref{eq:kkt1} and~\eqref{eq:kkt2} with
$\lambda = k/(k+\gamma)$.  It is therefore an equilibrium allocation
vector, and the value of~$\hat{x}_1$ is as claimed.\footnote{The
  corresponding equilibrium staking distribution~$\hat{\pi}$ is
  proportional to the equilibrium allocation vector~$\hat{\bfx}$, with the
  scaling factor set (as a function also of~$r$ and~$c$) so that the
  equilibrium conditions in~\eqref{eq:eq} are satisfied.}

For part~(a), consider an arbitrary $(\gamma,k)$-competitive set
of~$n$ BPs, with reward multipliers~$\mu_1 \ge \mu_2 \ge \cdots \ge
\mu_n = 1$.  Define alternative reward multipliers~$\mu'$ as follows:
$\mu'_1 = \mu_1$, $\mu'_i = \gamma \cdot \mu_1$ for
$i=2,3,\ldots,k+1$, and $\mu'_i = 0$ for $i=k+1,k+2,\ldots,n$.
Because the given BP set is~$(\gamma,k)$-competitive, $\mu'_i \le
\mu_i$ for every~$i$.  The First Monotonicity Lemma
(Lemma~\ref{l:mono1}), applied once for each~$i=2,3,\ldots,n$, implies
that $\hat{x}_1(\mu,r,c) \le \hat{x}_1(\mu',r,c)$. (That lemma holds
for arbitrary nonnegative reward multipliers and does not require
that~$\mu'_n = 1$.)  By the proof of part~(b), $\hat{x}_1(\mu',r,c)
= 1 - (\gamma \cdot k)/(k+\gamma)$. (Additional BPs with $\mu_i=0$
do not change the equilibrium allocation.) Thus $\hat{x}_1(\mu,r,c)
\le 1 - (\gamma \cdot k)/(k+\gamma)$, as required.
\end{prevproof}

\section{Long-Run Forces Toward Centralization}\label{s:prob}

\subsection{The Staking Process}

We next extend the basic model of Section~\ref{s:model} in an
orthogonal direction, to investigate connections between BP
heterogeneity and centralization from a different angle.  In contrast
to the single-shot setup of Section~\ref{s:econ}, this section studies
the evolution over time of the stakes controlled by different BPs. We
assume here that BPs are non-strategic and always stake all of the
native currency that they possess (including reinvesting any earned
rewards back into staking).\footnote{Combining the strategic features of the
model in Section~\ref{s:econ} with the long-run dynamics studied in
this section is an interesting direction for future research.}

Precisely, we again consider a set~$I$ of~$n$ BPs with reward
multipliers $\mu_1 \ge \mu_2 \ge \cdots \ge \mu_n=1$ and a base
reward~$r$.  In addition, each BP~$i$ has a positive initial stake,
denoted~$\pi_{i,1}>0$. More generally, $\pi_{i,t}$ will denote the stake
controlled by BP~$i$ after~$t-1$ block production opportunities have
passed.  We consider the following probabilistic process, which we
call the {\em staking process}, 
for~$t=1,2,\ldots,$:

\begin{itemize}

\item For block production opportunity~$t$, one BP is chosen with
  probability proportional to the current BP stakes. That is, BP~$i$
  is chosen  with probability
  \[
    x_{i,t} = \frac{\pi_{i,t}}{\sum_{j \in I} \pi_{j,t}}.
\]
    
\item If BP $i^*$ is chosen for this opportunity, then the stakes
  evolve according to:
          \begin{equation*}
            \pi_{i,t+1}=
            \begin{cases}
              \pi_{i,t} + \mu_{i} \cdot r, & \text{if}\ i=i^* \\
              \pi_{i,t}  & \text{if}\ i\neq i^*.
            \end{cases}
        \end{equation*}
\end{itemize}
Thinking of~$r$ as fixed throughout this section, the staking process
is fully described by the reward multiplier vector
$\mu=(\mu_1,\ldots,\mu_n)$ and the initial stake vector
$\pi=(\pi_{1,1},\ldots,\pi_{n,1})$.\footnote{Extending the analysis of
  this section to time-varying rewards is an interesting and
  challenging direction for future work.}
We generally assume (without loss
of generality) that the entries of~$\mu$ are sorted in nonincreasing
order; we make no assumptions about~$\pi$.
One advantage of this setup is that it connects directly to
the theory of generalized P\'olya urn models (e.g., \cite{polya_textbook}).

We will be interested in the long-run behavior of the staking process
(as a function of~$\mu$ and~$\pi$), and in particular whether it
``centralizes'' in the sense that a single BP dominates the staking
distribution:
\begin{definition}[$\eps$-Centralization]\label{d:eps}
For a parameter~$\eps > 0$, the staking process is {\em
  $\epsilon$-centralized at block $t$} if, for some BP~$i$,
\[
  x_{i,t} \ge 1 - \eps.
  \]
\end{definition}
With this setup and notion of centralization, we now have another
concrete way to quantify the extent to which heterogeneity across
BPs leads to centralization: as a function of the heterogeneity in the reward
multiplier sequence (the~$\mu_i$'s), how quickly (if at all) does the
staking process become $\eps$-centralized?

It is natural to conjecture that, if~$\mu_1 > \mu_2$ and
$t \rightarrow \infty$, the first BP should, with probability~1, end
up with all but a
vanishingly small fraction of the overall stake (no matter what the
initial stake distribution is).  (Intuitively, the rate of increase in
the first BP's stake should outpace that of the others.)  This
conjecture is in fact true, although the proof is
not at all obvious (see Janson~\cite{janson}). Our concern here will
be on the {\em speed} with which the staking distribution concentrates.
The hope is that, provided there is not too much heterogeneity across
BPs (e.g., with~$\mu_2$ close to~$\mu_1$), stake concentrates
slowly, perhaps even slowly enough to not be a first-order
concern. Such a quantitative analysis is necessary to assess the
potential benefits of any design aimed to reduce BP heterogeneity (and
slow down centralization), such as the proposer-builder separation
idea described in Section~\ref{s:pbs}.

\subsection{Quantifying the Speed of Centralization}

We are now primarily interested in defining quantitative lower and
upper bounds on how many blocks it takes for the staking process to
become $\epsilon$-centralized. For ease of exposition we will refer to
the sum of all the BPs' stakes at block~$t$, not including BP~1, by
$\pi_{-1,t}=\sum_{j\in I,j\neq 1} \pi_{j,t}$. We omit the subscript on
time when referring to the starting stakes, writing
$\pi_1$ for $\pi_{1,1}$ and $\pi_{-1}$ for $\pi_{-1,1}$. 

\begin{theorem}[Bounds on Number of Blocks for $\epsilon$-Centralization]\label{t:bounds}
Let $\beta = \max\{\frac{\mu_1r}{\pi_1},\frac{\mu_2r}{\pi_{-1}}\}$ and
$\rho = \frac{\pi_{-1}}{\pi_1}$. Then for every $\epsilon > 0$:
\begin{enumerate}

\item (Upper bound on time to centralization) For every
\[t >
  \frac{3}{2\mu_2r}\left(\pi_1\left(\frac{3\rho(1-\epsilon)}{\epsilon}\right)^{\frac{\mu_1}{\mu_1-\mu_2}}+\pi_{-1}\left(\frac{3\rho(1-\epsilon)}{\epsilon}\right)^{\frac{\mu_2}{\mu_1-\mu_2}}\right),
\]
  we have
  $$\Pr\left(x_{1,t} < 1-\epsilon \right) < 8\beta.$$

\item (Lower bound on time to centralization) For every
\[  
t<
  \frac{1}{2\mu_1r}\left(\pi_1\left(\frac{\rho(1-\epsilon)}{3\epsilon}\right)^{\frac{\mu_1}{\mu_1-\mu_n}}
    +
    \pi_{-1}\left(\frac{\rho(1-\epsilon)}{3\epsilon}\right)^{\frac{\mu_n}{\mu_1-\mu_n}}
    - 2(\pi_1+\pi_{-1})\right),
\]
we
  have
  $$\Pr\left(x_{1,t} \ge 1-\epsilon \right) < 8\beta.$$
    
    \end{enumerate}
\end{theorem}

The proof of Theorem~\ref{t:bounds} can be found in the Appendix.
Observe that the parameter~$\beta$ that controls the probability
bounds is small provided the rewards earned by BPs are small relative
to their initial stakes.  We can see that the dominating factor in the
speed to $\epsilon$-centralization is the difference between $\mu_1$
and $\mu_2$. Reducing this gap causes an exponential increase in the
time it takes for the process to $\epsilon$-centralize. As a concrete
example, 
suppose that BP~1 begins with 10\% of the stake and we consider the
time required for it to control at least 33\% of the stake. This
corresponds to the parameter values $\rho=9$ and $\epsilon=2/3$.
In this case, decreasing $\mu_1-\mu_2$ from $0.5$
to $0.1$ can increase the number of blocks needed to
$\epsilon$-centralize by roughly five orders of magnitude.
Echoing the economic analysis in Section~\ref{s:econ}, the
probabilistic analysis in this section underscores the importance of
mechanisms for eliminating large differences between the reward
multipliers of different BPs.

\section{Proposer-Builder Separation and BP Heterogeneity}\label{s:pbs}

\subsection{Idealized Proposer-Builder Separation}

Sections~\ref{s:econ} and~\ref{s:prob} take very different paths to
arrive at a similar conclusion: significant BP heterogeneity leads to
concentration of stake (at equilibrium with endogenous staking, or in
the long run with automatic reinvestment) and, conversely,
decentralization can largely survive a small amount of BP
heterogeneity.  But, in practice, why wouldn't there be a large amount
of BP heterogeneity? As discussed in Section~\ref{s:intro}, such
heterogeneity can easily arise from many sources, such as exclusive
access to certain transactions, proprietary block-building algorithms,
or differences in computing resources.  In the context of
Theorems~\ref{t:BP_comp} and~\ref{t:bounds}, one might expect that the
parameters $1/\gamma$ and $\mu_1-\mu_2$ would, in reality, be quite
large.

As discussed in Section~\ref{s:intro}, proposer-builder separation
(PBS) is one approach to reducing the degree of BP heterogeneity.
An extremely idealized version of PBS, chosen to indicate the
best-case scenario, works as follows:
\begin{itemize}

\item Every participant chooses exactly one of two roles: a builder or
  a proposer.

\item Only proposers participate directly in the blockchain protocol
  as validators (committing stake, proposing blocks, casting votes,
  etc.).

\item For each block production opportunity, each builder sends a
  block (along with a bid) to the proposer that has been chosen for
  that opportunity.  (A builder's block should be materially independent of
  the identity of the proposer.)

\item The block proposer proposes the block with the highest bid.

\item The block proposal is accepted by the other validators and
  finalized.

\item The winning builder pays its bid to the block proposer.

\end{itemize}
The key observation is that, under the above assumptions, there is no
longer any heterogeneity across proposers---the reward for a block
production opportunity is the same (namely, the highest bid submitted by a
builder), no matter which proposer is chosen to take advantage of it.
This case translates to~$\gamma=1$ in Theorem~\ref{t:BP_comp} (with
all proposers committing an equal fraction of stake at equilibrium)
and $\mu_1 = \mu_n$ in Theorem~\ref{t:bounds} (with the staking
distribution not necessarily concentrating at all), and in these
senses is consistent with sustainable decentralization.\footnote{There might
still be heterogeneity across builders, potentially leading
to centralization within the builder set. Builder centralization also
poses risks, such as censorship and collusion, but
a discussion of these is outside the scope of this paper.}

There are, of course, many ways that reality could depart from this
idealized scenario. Here, we stress-test the implicit assumption in
PBS that there is no overlap between the set of builders and the set
of proposers. (For example, perhaps a successful block-builder chooses
to invest some of its profits into operating a number of validators.)
Specifically, we allow proposers to privately engage in
block-building on the side, in the hopes of constructing a block even
more profitable than the highest bid by one of the specialized
builders.  At the extreme, if proposers are much better block-builders
than the specialized builders, proposer-builder separation achieves
nothing and we should expect significant proposer heterogeneity and
consequent centralization.  The hope, then, would be that the builder
ecosystem is sufficiently competitive that proposers can rarely if
ever profit from privately constructing their own blocks.  We next
develop a simple model of a competitive builder ecosystem and
formalize the extent to which this hope is in fact correct.

\subsection{The Equalizing Effects of Competing
  Builders}\label{ss:competition}

We focus on a fixed block production opportunity and consider a finite
set~$Y$ of~$k$ specialized builders.  Each builder~$y \in Y$ is
capable of building a block that generates for it a nonnegative reward
of~$r_y$. Variations in~$r_y$ across builders~$y$ could stem from
different information, different block-building algorithms, different
choices of random seeds, and so on.  We assume that the $r_y$'s are
independent draws from a distribution~$D_Y$---in this sense, the
builders in~$Y$ are all equally proficient on average (for example,
because all the inferior builders have already been competed away).
The competitiveness of the builder ecosystem is then most simply
measured by its size, $k$.

There is a separate finite set~$I$ of proposers.\footnote{For example,
  in the Ethereum ecosystem, the parameter~$k$ is generally thought of
  as small (e.g., 5) while the size of~$I$ would be in the thousands.}
One of these,
say~$i$, is chosen for the block production opportunity under
consideration.  The proposer~$i$ accepts blocks from builders (along
with bids), and optionally also builds its own block.  We assume that
the proposer is capable of building a block that generates a reward
of~$r_i$ for it, where~$r_i$ is drawn from a distribution~$D_i$.
We assume that proposers follow what is then the obvious
reward-maximizing strategy:
\begin{itemize}

\item [(S1)] Accept blocks from builders; let~$B^*$ denote the submitted
  block with the highest bid (from the builder to the proposer), and
  denote this bid by~$b^*$.

    \item [(S2)] Privately construct a block~$B$ that would generate a reward
      of $r_i \sim D_i$ for~$i$.

      \item [(S3)] Propose either the block~$B^*$ (if $b^* > r_i$) or~$B$
        (otherwise), thereby earning a reward of $\max\{b^*,r_i\}$.

\end{itemize}
Heterogeneity across proposers is captured in this model by the
proposer-specific distribution~$D_i$. Perhaps some proposers do not
have the resources or inclination for private block-building, and
always accept the best block submitted by a builder (so in effect,
$D_i$ is a point mass at~0) while others privately compete with the
specialized builders. Or perhaps some proposers have access to a
larger pool of pending transactions than others.  How much variation
in expected reward is there across proposers, and how does the answer
depend on the competitiveness~$k$ of the builder
ecosystem?\footnote{This setup assumes that block-building is costless
  and that a proposer engages in private block-building only when it
  is chosen for a block production opportunity. A more general version
  of the model would incorporate the costs of block-building (e,g., from
  maintaining positions on a centralized exchange to take advantage of
  arbitrage opportunities) and would allow a proposer to compete with
  the specialized builders for every block production
  opportunity.}

Some version of the following three assumptions is necessary to prove
bounds on the degree of heterogeneity of proposer rewards in this
model:
\begin{itemize}

\item [(A1)] By publishing a builder-proposed block, the proposer
  should earn 
  no reward beyond the bid of the winning builder. (Intuitively, the
  builder should have already extracted all of the value of the block
  it proposed.)  If this assumption does not hold, different proposers
  may earn rewards at much different rates even when they all follow
  the strategy~(S1)--(S3).

\item [(A2)] Proposers cannot be significantly better block-builders
  than the builders in~$Y$. (If they are, the builder ecosystem
  doesn't matter.) Formally, we will assume that every proposer
  distribution~$D_i$ is first-order stochastically dominated (FOSD) by the
  builder distribution~$D_Y$. (I.e., for all $x \ge 0$, $\prob[r \sim
  D_i]{r \ge x} \le \prob[r \sim D_Y]{r \ge x}$.)

\item [(A3)] Proposer distributions~$D_i$ cannot be excessively
  heavy-tailed. (Otherwise, a proposer with such a distribution could
  earn vastly more rewards in expectation than a proposer that does no
  private block-building, even though it is almost never able to
  outperform the specialized block-builders.  For example, this is
  true if~$D_i$ is the ``equal-revenue distribution,''
  with distribution function $1-\tfrac{1}{x}$ on $[1,\infty)$.) Formally, we will
  assume that the builder distribution~$D_Y$ (which FOSD every
  proposer distribution~$D_i$) satisfies the monotone hazard rate
  (MHR) condition, meaning that $f(x)/(1-F(x))$ is nondecreasing,
  where~$f$ and~$F$ denote the PDF and CDF of the
  distribution. Intuitively, the tails of an MHR distribution are no
  heavier than those of an exponential distribution.
  
\end{itemize}
The main result of this section shows that, under
assumptions~(A1)--(A3), a competitive builder ecosystem ensures
minimal proposer heterogeneity.
\begin{theorem}[Competition Reduces Proposer Heterogeneity]\label{t:pbs}
Let~$Y$ denote a set of~$k$ specialized builders, with rewards~$r_y$ drawn
i.i.d.\ from a distribution~$D_Y$ that satisfies the MHR
condition, and assume that builders bid according to the (unique)
Bayes-Nash equilibrium of a symmetric first-price auction with value distribution~$D_Y$.
Let~$I$ denote a set of proposers, with proposer~$i$'s
private block-building reward~$r_i$ drawn from a distribution~$D_i$
that is FOSD by~$D_Y$. Assume that every proposer follows the strategy
in~(S1)--(S3).  Then, for every pair~$i,j \in I$ of proposers, the
expected reward earned by~$i$ (conditioned on selection) is at most
\[
1 + O\left(\frac{1}{\log k}\right)
\]
times that earned by~$j$ (conditioned on selection).
\end{theorem}
Thus, as the competitiveness~$k$ of the builder ecosystem increases,
the ratio in expected proposer rewards (roughly corresponding to the
parameter~$\gamma$ in Theorem~\ref{t:BP_comp} or~$\mu_1$ in
Theorem~\ref{t:bounds}) tends to~1.

\vspace{.1in}

\begin{prevproof}{Theorem}{t:pbs}
The minimum expected reward (conditioned on selection) would be earned
by a proposer that never engages in private block-building
(equivalently, $D_i$ is a point mass at~0). The expected reward of
such a proposer would be the expected highest bid by a builder. It is
known that, at the unique
Bayes-Nash equilibrium\footnote{In a Bayes-Nash equilibrium, each
  player~$i$ picks a strategy (i.e., a mapping of each potential
  reward~$r_i$ to a bid~$b_i$) that, assuming other players bid
  according to their equilibrium strategies, always maximizes its expected
  revenue (i.e., its bid times its probability of winning).} of a
first-price auction with i.i.d.\ private valuations (i.e., $r_i$'s)
and~$k$ bidders, the expected highest bid is the expected
second-largest sample of~$k$ i.i.d.\ samples from the value
distribution (i.e., from~$D_Y$); see e.g.~\cite[Proposition
2.3]{krishna}.\footnote{This quantity is easily seen to be the
  expected revenue of a second-price (Vickrey) auction in this
  scenario; the stated fact then follows from revenue equivalence
  (using that first- and second-price auctions have the same
  allocation rule at equilibrium, and that the revenue of a first-price
  auction equals the highest bid).}
Call this quantity~$R$.  How much bigger could the expected reward
(conditioned on selection) of a different proposer be?

Fix a proposer~$i$ with reward distribution~$D_i$ that is FOSD
by~$D_Y$. In the notation of the strategy~(S1)-(S3), the expected
reward of this proposer (conditioned on selection) is the expected
value of $\max\{b^*,r_i\}$. Because builders will, at equilibrium, bid
at most their values in a first-price auction, this expected reward
can be bounded above by the expectation of
$\max\{ \max_{y \in Y} r_y, r_i \}$---the expected maximum out of~$k$
i.i.d.\ samples from~$D_Y$ and one independent sample from~$D_i$.
Because~$D_Y$ FOSD $D_i$, this quantity can, in turn, be bounded above
by the expected maximum of~$k+1$ i.i.d.\ samples from~$D_Y$.

Next, we use the fact that, because~$D_Y$ satisfies the MHR condition,
the expectation of the second-largest value out of $n \ge 2$ i.i.d.\
samples from~$D_Y$ is a concave function of~$n$
(see~\cite[Theorem~1(b)]{watt}).  Because this quantity is also
nonnegative and nondecreasing in~$n$, this fact implies that the
expected second-largest value of~$k+1$ i.i.d.\ samples from~$D_Y$ is
at most $1+\tfrac{1}{k-1}$ times the expected second-largest value
of~$k$ i.i.d.\ samples from~$D_Y$ (i.e., at most
$(1+\tfrac{1}{k-1})R$).

It remains to bound the factor by which the expected value of the
largest of~$k+1$ i.i.d.\ samples from~$D_Y$ exceeds that of
the second-largest.  Using the representation
\[
F(x) = 1 - \exp\left\{ - \int_{0}^x h(x)dx\right\}
\]
of a distribution function~$F$ with density~$f$, where~$h(x) =
f(x)/(1-F(x))$ denotes its hazard rate, it follows that this factor is
maximized among MHR distributions (those with~$h$ nondecreasing) by
exponential distributions (those with~$h$ constant). A calculation
then shows that this factor is at most
\[
\frac{\sum_{j=1}^{k+1} \tfrac{1}{j}}{\sum_{j=2}^{k+1} \tfrac{1}{j}} =
1 + O\left( \frac{1}{\log k} \right).
\]

This completes the proof: the expected reward (conditioned on
selection) of the proposer~$i$ is at most the expected value of
$\max\{\max_{y \in Y} r_y, v_i\}$, which is at most the expected value
of the largest of~$k+1$ i.i.d.\ samples from~$D_Y$, which is at
most~$1+O(1/\log k)$ times the expected value of the second-largest
of~$k+1$ i.i.d.\ samples from~$D_Y$, which is at
most~$(1+1/(k-1))(1+O(1/\log k)) = 1+O(1/\log k)$ times~$R$ (i.e.,
the expected second-largest of~$k$ i.i.d.\ samples from~$D_Y$).
\end{prevproof}

\subsection*{Acknowledgments}

We thank Justin Drake, Ben Fisch, Mike Neuder, Matt Weinberg, and the
FC reviewers for useful comments on a preliminary version of this
paper.

\bibliographystyle{splncs04}
\bibliography{bibliography}

\appendix

\section{Proof of Lemma \ref{l:mono1}}

\begin{proof}
The optimality conditions for the concave optimization problem in
Proposition~\ref{prop:jt} state that optimal solutions are those that
are feasible and that equalize the partial derivatives of the
objective function (except for unused variables, for which
the partial derivative may be lower). That is, a feasible allocation
$\bfx$ is optimal for~\eqref{eq:opt1}--\eqref{eq:opt3}
if and only if there is a constant~$\lambda$
such that
\begin{equation}\label{eq:kkt1}
\frac{\partial \Phi_i}{\partial x_i}(\bfx) = \mu_i \cdot \frac{r}{c}
\cdot (1 - x_i) \le \lambda
\end{equation}
for all $i \in I$ and
\begin{equation}\label{eq:kkt2}
\frac{\partial \Phi_i}{\partial x_i}(\bfx) = \mu_i \cdot \frac{r}{c}
\cdot (1 - x_i) = \lambda
\end{equation}
for all~$i$ with~$x_i > 0$, where~$\Phi_i$ denotes the $i$th term of
the objective function in~\eqref{eq:opt1}.  Starting from the
allocation~$\hat{\bfx}(\mu,r,c)$ (with constant~$\lambda$),
increasing~$\mu_j$ increases $\partial \Phi_j/\partial x_j$.  If its
new value remains at most~$\lambda$, then the equilibrium allocation
remains the same.  If its new value exceeds $\lambda$, to compensate
and restore the optimality conditions, $x_j$ must increase and every
other~$x_i$ (that is not already~0) must decrease.
\end{proof}

\section{Proof of Theorem \ref{t:bounds}}

\subsection{Proof of Upper Bound}
We start with a monotonicty lemma that shows
increasing how competitive the block producers are can only increase
the number of blocks it takes for centralization to occur. This is
intuitive since if BP 1's competitors start earning more, the rate at
which BP 1 accrues their advantage becomes slower. We will use this to
show that if $\mu'$ is more competitive than $\mu,$ then an upper
bound on the number of blocks till $\epsilon$-centralization for
$\mu'$ also holds for $\mu$.

\begin{lemma}[Second Monotonicity Lemma] \label{l:monotone}
Let $\mu$ and $\mu'$ be two block producer sets with initial starting
stakes $\pi$ where $\mu_i'\ge \max_{i\in I, i\neq 1}\{\mu_i\}$ for all $i \ge 2$. Then $\Pr(x_{1,t} > a) \ge \Pr(x'_{1,t} > a)$ for all $a\in[0,1]$ where $x_{1,t}$ and $x'_{1,t}$ correspond to BP 1's chance of being chosen at block $t$ under $(\mu,\pi)$ and $(\mu',\pi)$ respectively. 
\end{lemma}

\begin{proof}
    Here we give a formal proof of the coupling argument sketched in the main body. Let $\Pi=(\pi_t)_{t\ge 1}$ and $\Pi'=(\pi'_t)_{t\ge 1}$ be the staking processes under $\mu$ and $\mu'$ respectively, both with starting stakes $\pi$. Then consider the following coupling for $(\Pi,\Pi')$. 

    For all blocks $t$: Define $a_{0,t}=a'_{0,t}=0$. Then for all $i\in I$, recursively define $$a_{i,t}=x_{i,t}+a_{i-1,t} \text{ and similarly }a'_{i,t}=x'_{i,t}+a'_{i-1,t}$$

    Then consider the following process to choose the block producer for every block $t$:\\ (1) Sample $y_t\sim U[0,1]$, (2) Let $i^*_t=i$ if $y_t\in (a_{i-1,t},a_{i,t}]$, (3) Let $i'^*_t = i$ if $y_t\in (a_{i-1,t},a_{i,t}]$

    Crucially we are using the same randomness to sample the winner under both $\Pi$ and $\Pi'$ in every block.  Combining this with the fact $a_{i,t}-a_{i-1,t}=x_{i,t}$, $a'_{i,t}-a'_{i-1,t}=x'_{i,t}$ and $y_t$ is sampled uniformly random from $[0,1]$ gives us that the marginal distributions of $\Pi$ and $\Pi'$ under this coupling match their distributions under the staking process. Thus this process is a valid coupling, and to prove the lemma, it suffices to show that  for every sequence of random draws $\{y_t\}_{t\ge 1}$ that $x_{1,t} \ge x'_{1,t}$. We proceed to show this via strong induction. \\ 

    The base case follows trivially since both processes start with the same stakes $\pi$ so $x_{1,1}=x'_{1,1}$. Now for the inductive hypothesis assume that for all sequences $\{y_t\}_{1\le t\le k-1}$, $x_{1,s} \ge x'_{1,s}$ for all $s=1,...,k$ . We then show that for any $y_k\in[0,1]$ that $x_{1,k+1} \ge x'_{1,k+1}$. 
    
    We start by claiming that $\pi_{1,k} \ge \pi'_{1,k}$. This follows from the strong inductive hypothesis since $x_{1,s} \ge x'_{1,s}$ for all $s\le k$ implies that whenever $y_t  \le x'_{1,t}$ that $y_t\le x_{1,t}$ for all $t\le k$. Since $\mu_1=\mu'_1$ this gives us $\pi_{1,k}\ge \pi'_{1,k}$. Now fix any $y_{k}\in [0,1]$. Then from our inductive hypothesis, we  have 3 cases: (1) $y_k \le x'_{1,k} \le x_{1,k}$, (2) $x'_{1,k} < y_k < x_{1,k}$, or (3) $x'_{1,k} < x_{1,k}< y_k$. In particular note the case where $i^*_k\neq 1$ but $i'^*_k=1$ can never happen. We will show that each of these three cases will lead to $x_{1,k+1}\ge x'_{1,k+1}$. We use $\Pi_t = \sum_{i\in I} \pi_{i,t}$ and $\Pi'_t = \sum_{i\in I} \pi'_{i,t}$ to refer to the total stakes at time $t$ under $\Pi$ and $\Pi'$ respectively. 

    \begin{enumerate}
        \item $i^*_k=1, i'^*_k=1$: Here  $\frac{\pi_{1,k}}{\Pi_k} \ge \frac{\pi'_{1,k}}{\Pi'_k} \implies \frac{\pi_{1,k}+\mu_1r}{\Pi_k+\mu_1r} \ge \frac{\pi'_{1,k}+\mu_1r}{\Pi'_k+\mu_1r}$
        \item $i^*_k=1, i'^*_k=i\neq1$: Here $\frac{\pi_{1,k}}{\Pi_k} \ge \frac{\pi'_{1,k}}{\Pi'_k} \implies \frac{\pi_{1,k}+\mu_1r}{\Pi_k+\mu_1r} \ge \frac{\pi'_{1,k}}{\Pi'_k+\mu_ir}$ 
        \item $i^*_k=i\neq 1, i'^*_k=j\neq 1$: Here we use $\pi_{1,k} \ge \pi'_{1,k}$, $\mu_i \le \mu'_j$, and $\frac{\pi_{1,k}}{\Pi_k} \ge \frac{\pi'_{1,k}}{\Pi'_k}$ to get $\frac{\pi_{1,k}}{\Pi_k+\mu_i} \ge\frac{\pi'_{1,k}}{\Pi'_k+\mu'_j} $
    \end{enumerate}

    Hence in all cases we get $x_{1,k+1}\ge x'_{1,k+1}$. 
    
\end{proof}

With the monotonicity of the effect of competitiveness on the blocks to centralization established, we are ready to prove the main theorem.

\begin{proof}
    We start by proving the upper bound for the staking process $(\mu,\pi)$. 

    \paragraph{Reducing to Two Block Producers:} We first show that it suffices to prove the upper bound for the case with two block producers where $\hat{\mu} = (\mu_1,\mu_2)$ and $\hat{\pi}=(\pi_1,\pi_{-1})$. To see this, consider $\mu' = (\mu_1,\mu_2,...,\mu_2)$ where $\mu'_i=\mu_2$ for $i\ge 2$ with the same starting stakes $\pi$. Since $\mu_2 =\max_{i\in I, i\neq 1}\{\mu_i\}$ , by lemma \ref{l:monotone} we have that the probability $(\mu,\pi)$ is $\epsilon$-centralized at block $t$ first order stochastically dominates the probability  $(\mu',\pi)$ is $\epsilon$-centralized at block $t$. Furthermore, note that for the purposes of finding an upper bound on the number of blocks it takes for $(\mu',\pi)$ to $\epsilon$-centralize, we are only concerned with how fast $\pi'_{1,t}$ grows compared to $\pi'_{-1,t}$. Then since $\mu'_i=\mu_2$ for all $i\ge 2$, $\pi'_{-1,t}$ increases by the same amount if $i^*_t\neq 1$ regardless of which block producer is chosen. Thus the probability $i^*_t\neq 1$ for future blocks is also the same regardless of which $i\in I\setminus \{1\}$ is chosen. This implies that the time to $\epsilon$-centralization is equivalent under $(\mu',\pi)$ and $(\hat{\mu},\hat{\pi})$. It follows that $\Pr(\hat{x}_{1,t} > a) = \Pr(x'_{1,t} > a) \le \Pr(x_{1,t} > a)$ for all blocks $t\ge 1$ and  $a\in [0,1]$, making it sufficient to prove the upper bound holds for $(\hat{\mu},\hat{\pi})$. Thus from here we will assume WLOG that there are only two block producers with $\mu = (\mu_1,\mu_2)$ and $\pi=(\pi_1,\pi_{-1})$
 
    \paragraph{Passing to the Continuous Case:} As has been noted previously, the process describing how stake evolves can be modeled as a P\'olya urn with a diagonal replacement matrix. Imagine an urn with all of the staked coins where each coin has a corresponding block producer. Then for every block production opportunity, a coin is randomly sampled from the urn. If the coin is owned by BP $i$, then $\mu_i r$ additional coins corresponding to BP $i$  are added to the urn. Analyzing the dynamics of \emph{unbalanced} urns, where different amounts of coins are added to the urn depending on who is chosen, tends to be difficult in the discrete scheme because the total amount of coins in the urn is dependent on the history of what coins were chosen.  
    
    To get around this, a classic technique in studying P\'olya urns is embedding the discrete process into a continuous one. In the continuous case, the discrete sampling process is replaced with a birth-death branching process. Every coin in the urn is given a clock with a i.i.d Exp(1) random expiry date. If a coin belonging to BP $i$ expires, that coin's clock resets and $\mu_i r$ new coins belonging to BP $i$ are added to the urn, each also with their own i.i.d. Exp(1) clocks. Since the Exp(1) distribution is memoryless, it follows that at any given time, the next clock to go off will be uniformly sampled from all of the coins in the urn. Thus the probability the next coin chosen is BP $i$'s matches the probability they would have been selected in the discrete process. Thus, if we consider snapshots of the continuous process at the times the clocks expire, then this process is distributed identically to the discrete process. This lets us prove results about the staking process by proving results in the continuous version of the staking process and then inverting back to the discrete case. 

The primary advantage of working in the continuous case is that coins from different BPs don't affect each other. Notice that a clock expiring for one BP's coins have no influence on how another BP's coins evolve. We can make this precise by observing that each BP's coins grow according to independent Yule processes. Since we have shown it suffices to consider the case with two block producers, we will have two processes $\Pi_1(t)$ and $\Pi_2(t)$ representing the evolution of stake for $\pi_{1,t}$ and $\pi_{2,t}$ respectively. The fact that $\Pi_1(t)$ and $\Pi_2(t)$ evolve independently, lets us bound the ratio of $\Pi_1(t)/\Pi_2(t)$ via concentration bounds on $\Pi_1(t)$ and $\Pi_2(t)$ independently. From \cite{polya_textbook} we have that E$[\Pi_1(t)]=\pi_1e^{\mu_1rt}$ and Var$[\Pi_1(t)]=\mu_1r\pi_1(e^{2\mu_1r}-\mu_1rte^{\mu_1rt})\le \mu_1 r\pi_{1}e^{2\mu_1rt}$. Thus we can apply Chebyshev's inequality to get 

     \begin{align*}
         \Pr\left(|\Pi_1(t) - E[\Pi_1(t)] |\ge \frac{1}{2}E[\Pi_1(t)]\right) \le \frac{4 Var[\Pi_1(t)]}{E[\Pi_1(t)]^2} \\ 
         \implies \Pr\left(\Pi_1(t) < \frac{1}{2}E[\Pi_1(t)]\right) \le \frac{4\mu_1r\pi_{1}e^{2\mu_1rt}}{\pi^2_{1}e^{2\mu_1rt}} = \frac{4\mu_1 r}{\pi_{1}} 
     \end{align*}

    The same derivation for BP 2 gives $$\Pr\left(\Pi_2(t) \ge \frac{3}{2} E[\Pi_2(t)]\right) \le \frac{4\mu_2 r}{\pi_{-1}}. $$

    Then using a union bound gives $$\Pr\left(\frac{\Pi_1(t)}{\Pi_2(t)} < \frac{\frac{1}{2}E[\Pi_1(t)]}{\frac{3}{2}E[\Pi_2(t)]}\right) <  4r\left(\frac{\mu_1}{\pi_{1}}+\frac{\mu_2}{\pi_{-1}}\right) $$

    Now let $\hat{t}=\frac{\ln \frac{3\pi_{-1}}{\pi_1}\frac{1-\epsilon}{\epsilon}}{r(\mu_1-\mu_2)}$. This gives us 

    \begin{equation*}
        E[\Pi_1(\hat{t})]= \pi_1\left(\frac{3\pi_{-1}(1-\epsilon)}{\pi_1 \epsilon}\right)^{\frac{\mu_1}{\mu_1-\mu_2}} \text{ and } E[\Pi_2(\hat{t})]= \pi_{-1}\left(\frac{3\pi_{-1}(1-\epsilon)}{\pi_1 \epsilon}\right)^{\frac{\mu_2}{\mu_1-\mu_2}}
    \end{equation*}

    Plugging this into our union bound above gives for any $t\ge \hat{t}$, $$\Pr\left(\frac{\Pi_1(\hat{t})}{\Pi_2(\hat{t})} < \frac{1-\epsilon}{\epsilon}\right) \le 4r\left(\frac{\mu_1}{\pi_{1}}+\frac{\mu_2}{\pi_{-1}}\right)$$

    where $\frac{\Pi_1(\hat{t})}{\Pi_2(\hat{t})} < \frac{1-\epsilon}{\epsilon} \Leftrightarrow \frac{\Pi_1(\hat{t})}{\Pi_1(\hat{t})+\Pi_2(\hat{t})} < 1-\epsilon$
        
    \paragraph{Moving Back to the Discrete Case:} Now that we have shown a $\hat{t}$ for which the continuous process $\epsilon$-centralizes with constant probability, it remains to find how many blocks in the discrete process this corresponds to. Recall every time a clock expires in the continuous process, a block passes in the discrete process. Thus to get an upper bound on how many blocks have passed by time $\hat{t}$, it suffices to upper bound how many clocks are expected to expire by time $\hat{t}$. Observe that the number of coins must increase by at least $\mu_2 r$ everytime a clock goes off. So if we have an upper bound of the total amount of coins there are at time $\hat{t}$, we know the maximum number of clocks that could have expired. Using this gives an upper bound of  $\frac{\Pi_1(\hat{t})+\Pi_2(\hat{t})}{\mu_2r}$ blocks.

    Since Chebyshev's inequality is two-sided, with our union bound from before, we have already accounted for the probability that both $\Pi_1(\hat{t})$ and $\Pi_2(\hat{t})$ are above $3/2$ their expectations. Hence we can use this bound without suffering any loss to our probability, giving us a bound of

    {
    \setlength{\jot}{-3pt}
    \begin{align*}
    \frac{\frac{3}{2}(E[\Pi_1(\hat{t})]+E[\Pi_2(\hat{t})])}{\mu_2r} =   \frac{3}{2\mu_2r}\left(\pi_1\left(\frac{3\pi_{-1}(1-\epsilon)}{\pi_1 \epsilon}\right)^{\frac{\mu_1}{\mu_1-\mu_2}} \right. \nonumber 
    + \left. \pi_{-1}\left(\frac{3\pi_{-1}(1-\epsilon)}{\pi_1 \epsilon}\right)^{\frac{\mu_2}{\mu_1-\mu_2}}\right)
    \end{align*}
    }

Putting it all together, due to the equivalence of the continuous
process and the discrete process this gives us for $t$ greater than
this bound that $$\Pr\left(\frac{\pi_{1,t}}{\pi_{1,t}+\pi_{-1,t}} <  1-\epsilon\right) <4r\left(\frac{\mu_1}{\pi_{1}}+\frac{\mu_2}{\pi_{-1}}\right) <8\beta$$.\\ 
\end{proof}

\vspace{-1cm}

\subsection{Proof of Lower Bound}
Here we go through the  proof of the lower bound in Theorem \ref{t:bounds} for a staking process $(\mu,\pi)$. We largely elide the exposition justifying steps that are directly mirrored in the upper bound.

We start with a corollary of lemma \ref{l:monotone} that gives us a monotonicity result in the other direction. If we decrease the competitiveness of the BPs apart from BP
1, then the time to centralization can only decrease. Similarly to the
previous lemma, we will use this to show that if $\mu'$ is less
competitive than $\mu$, then a lower bound for the number of blocks
till $\epsilon$-centralization for $\mu'$ also holds for $\mu$.

\begin{corollary}\label{c:lower}
    Let $\mu$ and $\mu'$ be two block producer sets with initial starting stakes $\pi$ where $\mu'_i\le \min_{i\in I}\{\mu_i\}$ for all $i \ge 2$.
    Then $\Pr(x_{1,t} < a) \ge \Pr(x'_{1,t} <  a)$ for all $a\in[0,1]$ where $x_{1,t}$ and $x'_{1,t}$ correspond to BP 1's chance of being chosen at block $t$ under $(\mu,\pi)$ and $(\mu',\pi)$ respectively.   
\end{corollary}

\begin{proof}
    From how $\mu'$ is defined we have that $\mu_i\ge \max_{i\in I, i\neq 1}\{\mu'_i\}$. Thus by lemma \ref{l:monotone} we have that $\Pr(x'_{1,t} > a) \ge \Pr(x_{1,t} >  a)$ implying $\Pr(x_{1,t} < a) \ge \Pr(x'_{1,t} <  a)$. 
\end{proof}

With this monotonocity result we can continue with the proof of the lower bound.

\begin{proof}
    We start by reducing our analysis to the case of two block producers. Let $\mu'=(\mu_1,\mu_n,...,\mu_n)$. Then since $\mu_n = \min_{i\in I}\{\mu_i\}$ we can apply Corollary \ref{c:lower} to get that a lower bound for $(\mu',\pi)$ also holds for $(\mu,\pi)$. Then, like in the upper bound, we note that $\mu_i=\mu_j$ for all $i,j\neq1$ implies that $(\mu',\pi)$ $\epsilon$-centralizes at the same rate as the process $((\mu_1,\mu_n),(\pi_1,\pi_{-1})$. Thus WLOG we continue by letting $\mu=(\mu_1,\mu_n)$ and $\pi=\pi_1,\pi_{-1}$.

    Now we move to the continuous case, letting $\Pi_1(t)$ and $\Pi_2(t)$ represent the evolution of $\pi_{1,t}$ and $\pi_{-1,t}$ respectively. Note these are the same yule processes we had in the upper bound with the slight modification of changing $\mu_2$ to $\mu_n$ in $\Pi(2)$. Thus we can use Chebyshev again to get 
    $$\Pr\left(\Pi_1(t) > \frac{3}{2}E[\Pi_1(t)]\right) < \frac{4\mu_1r}{\pi_1} \text{ and } \Pr\left(\Pi_2(t) < \frac{1}{2}E[\Pi_2(t)]\right) < \frac{4\mu_n r}{\pi_{-1}}$$

    A union bound then gives $$\Pr\left(\frac{\Pi_1(t)}{\Pi_2(t)} > \frac{\frac{3}{2}E[\Pi_1(t)]}{\frac{1}{2}E[\Pi_2(t)]}\right) < 4r\left(\frac{\mu_1}{\pi}+\frac{\mu_n}{\pi_{-1}}\right)$$

    We can then let $\hat{t}=\frac{\ln \frac{\pi_{-1}}{3\pi_1} \frac{(1-\epsilon}{\epsilon})}{r(\mu_1-\mu_n)}$ giving us 
    \begin{equation*}
        E[\Pi_1(\hat{t})]= \pi_1\left(\frac{\pi_{-1}(1-\epsilon)}{3\pi_1 \epsilon}\right)^{\frac{\mu_1}{\mu_1-\mu_n}} \text{ and } E[\Pi_2(\hat{t})]= \pi_{-1}\left(\frac{\pi_{-1}(1-\epsilon)}{3\pi_1 \epsilon}\right)^{\frac{\mu_n}{\mu_1-\mu_n}}
    \end{equation*}

    Plugging $\hat{t}$ into the union bound then gives $$\Pr\left(\frac{\Pi_1(\hat{t})}{\Pi_2(\hat{t})} \ge \frac{1-\epsilon}{\epsilon}\right) \le 4r\left(\frac{\mu_1}{\pi_{1}}+\frac{\mu_n}{\pi_{-1}}\right)$$.

    Now all that remains is to convert back to a bound in the discrete case. Here we want a lower bound on the total amount of stake that the continuous process adds by time $\hat{t}$. Combining this with the fact that the maximum the total stake increases in any block is $\mu_1$ gives us a lower bound on how many blocks need to pass before the staking process can centralize. One difference in the analysis here between the upper and lower bound is that we also need to subtract the starting stake amounts on the bound for the total stake at time $\hat{t}$ to get a lower bound on how much stake was added. 

    Again like the upper bound, because Chebyshev is two-sided we have already paid for the probability bounding the probability $\Pi_1(t)$ and $\Pi_2(t)$ are below half their expectations. Thus without losing any more probability we can get a bound of 
    
    \begin{align*}
        \Pi_1(\hat{t})+\Pi_2(\hat{t}) < \frac{1}{2}\left(E[\Pi_1(\hat{t})]+E[\Pi_2(\hat{t})]\right) =  \frac{1}{2}\left(\pi_1\left(\frac{\rho(1-\epsilon)}{3\epsilon}\right)^{\frac{\mu_1}{\mu_1-\mu_n}}+\pi_{-1}\left(\frac{\rho(1-\epsilon)}{3\epsilon}\right)^{\frac{\mu_n}{\mu_1-\mu_n}}  \right)
    \end{align*}

Putting everything together gives a lower bound of 

\begin{align*}
    \frac{\frac{1}{2}E[\Pi_1(\hat{t})+\Pi_2(\hat{t})]-\Pi_1(0)-\Pi_2(0)}{\mu_1r} > 
\frac{1}{2\mu_1r}\left(\pi_1\left(\frac{\rho(1-\epsilon)}{3\epsilon}\right)^{\frac{\mu_1}{\mu_1-\mu_n}}
  +
  \pi_{-1}\left(\frac{\rho(1-\epsilon)}{3\epsilon}\right)^{\frac{\mu_n}{\mu_1-\mu_n}}
  - 2(\pi_1+\pi_{-1})\right)
\end{align*}

    where for $t$ less than this bound, $$\Pr(x_{1,t} \ge 1-\epsilon) < 4r\left(\frac{\mu_1}{\pi}+\frac{\mu_n}{\pi_{-1}}\right) <8\beta$$

\end{proof}

\end{document}